\documentclass[twocolumn]{IEEEtran}

\usepackage{amsmath, amssymb, graphicx, algorithm, algorithmic, color}

\usepackage{verbatim}
\usepackage{stfloats}
\usepackage{amsthm}
\usepackage{svg} 
\usepackage{amsmath}

\newtheorem{lemma}{Lemma}

\begin{document}

\title{Deep Q-Learning Based Resource Allocation in Interference Systems With Outage Constraint}

\author{Saniul Alam, 
        Sadia Islam,
        Muhammad R. A. Khandaker,~\IEEEmembership{Senior Member,~IEEE,} 
        Risala T. Khan,~\IEEEmembership{Senior Member,~IEEE,}
        Faisal Tariq,~\IEEEmembership{Senior Member,~IEEE},
        and Apriana Toding


}

\maketitle

\begin{abstract}
 This correspondence considers the resource allocation problem in wireless interference channel (IC) under link outage constraints. Since the optimization problem  is non-convex in nature, existing approaches to find the optimal power allocation are computationally intensive and thus practically infeasible. Recently, deep reinforcement learning has shown promising outcome in solving non-convex optimization problems with reduced complexity. In this correspondence, we utilize a deep $Q$-learning (DQL) approach which interacts with the wireless environment and learns the optimal power allocation of a wireless IC while maximizing overall sum-rate of the system and maintaining reliability requirement of each link. We have used two separate deep $Q$-networks to remove the inherent instability in learning process. Simulation results demonstrate that the proposed DQL approach outperforms existing geometric programming based solution.

\end{abstract}


%
\IEEEpeerreviewmaketitle

\section{Introduction}
Resource allocation tasks such as optimal power allocation has a significant effect on the capacity and performance of wireless communication \cite{jrnl_inter}. 
Traditionally, resource allocation problems have been approached and tackled by established numerical optimization approaches such as interference pricing \cite{wang2008price}, water filling  algorithms \cite{hu2011adaptive}, fractional programming \cite{shen2018fractional} and weighted minimum mean-squared error (WMMSE) minimization \cite{wmmse}. Most of these algorithms are mathematically formulated to optimize a certain system performance metrics such as sum-rate, minimum rate or mean-squared error. 
These algorithms generally adopted iterative approaches where they take certain parameters like channel realization as input and produce the optimum resource allocation strategy. 

Even though these established resource and power allocation methods have improved the performance to some extent, they come with their own set of problems. These non-convex optimization problems are NP-hard and thus incur high computational cost.  As a result, many of the algorithms and solutions become infeasible for real world deployment.
Moreover, as it takes time to compute the optimum allocation, it also has negative effect on the latency performance of the system which is a key performance metric for current and future communication systems \cite{jrnl_6g}. Despite the high computational cost, these algorithms have to be recomputed frequently since the channel is time varying and the allocation  may not remain usable for long. 

Machine  learning (ML) techniques has recently been considered in wireless communication research as an attractive tool for providing viable solution to major challenges such as channel decoding and estimation \cite{conf_v2i_ch_estim} as well as optimal power allocation \cite{JGao19}. 


Deep reinforcement learning (DRL), a member of ML family, has been considered as an attractive method which can learn the optimal solution of the resource allocation problem by interacting with the wireless environment. 
Reinforcement learning (RL) generally outperforms traditional approaches in most of the cases \cite{8382166}. The work in \cite{ahmed2019deep}  uses deep $Q$-learning (DQL) for centralized downlink power allocation scheme which maximizes the network throughput in a multi cell network. The work in \cite{nasir2019multi} also uses multi agent DQL to find optimal power allocation in a wireless cellular network, where each transmitter receives channel state information (CSI) from several neighbours and learns to adjust its transmit power accordingly. A deep $Q$-network (DQN) was used to overcome the instability problem of DL. In \cite{9295396}, the authors proposed faded-experience trust-region policy optimization (FE-TRPO) base power allocation algorithm which exploits continuous DRL. The work demonstrated that compared to the traditional WMMSE techniques, it significantly decreases computational complexity while maintaining similar performance. 

While most of these works use DRL for maximizing the overall system throughput, the schemes do not always guarantee reliability for each individual user. As a consequence, users with low channel gain may not be able to transmit at all while user with the best channel gain may be allowed to transmit at maximum power resulting in unfairness in the system. 
To address this limitation, this work exploits DQN with a reliability outage requirement threshold for each individual user. The DQN finds the relationship between available channel information and the solution of the power allocation problem. A DQN agent is then used to find the collective power allocation policy of the users that maximizes the overall system  throughput while maintaining  reliability of each user.
The main contribution of this paper can be summarised as:
\begin{itemize}
    \item Firstly, the resource allocation problem of a wireless interference network under probabilistic constraints has been solved using geometric programming (GP) (following \cite{boyd_out}).
    \item Then, a DRL-based solution has been developed in order to reduce computational complexity and make the solution practicable. Two separate deep $Q$-networks have been used in the DRL solution to remove the inherent instability of the learning process.
    \item Simulation results demonstrate that the proposed DRL approach significantly outperforms the GP approach.

\end{itemize}


\section{System Model}\label{sec_vision}
We consider an interference channel (IC) as shown in Fig.~\ref{rl_fig}, where $K$ single-antenna transmitters are communicating with $K$ single-antenna receivers. The direct channel between the $k$th transmitter-receiver pair is denoted by $h_{k,k}$ while $h_{k,j}$ denotes the IC between transmitter $j$ and receiver $k$. Let us denote $x_k$ as the signal transmitted by transmitter $k$ and $y_k$ as the corresponding received signal at receiver $k$. Then $y_k$ is given by
\begin{align}
  y_{k}=h_{k k} x_{k}+\sum_{j \neq k} h_{k j} x_{j}+n_{k}.
\end{align}
Here $n_{k} \sim \mathcal{N}\left(0, \sigma_k^{2}\right)$ for $\sigma_k > 0$ represents the additive channel noise. We assume that $p_k$ is the transmission power of transmitter $k$. Note that the symbol transmitted by different transmitters are independent of each other. Then the signal to interference-plus-noise ratio (SINR) can be expressed as 
\begin{align}
 \gamma_{k} \triangleq \frac{\left|h_{k k}\right|^{2} p_{k}}{\sum_{j \neq k}\left|h_{k j}\right|^{2} p_{j}+\sigma_{k}^{2} }\label{sink},
\end{align}
where $\sigma^2_{k}$ denotes the noise power at receiver $k$. Thus the achievable rate at the $k$th receiver under Gaussian channel fading is given by
\begin{align}
R_k = \log_2 \left(1+\frac{\left|h_{kk}\right|^{2} p_{k}}{\sum_{j\ne k}\left|h_{kj}\right|^{2} p_{j}+\sigma_{k}^{2}}\right).
\end{align}

\subsection{Problem Formulation}
Our aim is to maximize the overall system throughput of the communication links by optimally allocating transmit power among the active links such that the worst-user's mutual-information is maximized under link outage constraints. Accordingly, we formulate the following optimization problem:
\begin{subequations}\label{prob_msr0}
\begin{align}
\max_{\{p_k\}} ~&~ \min_{k} ~~ R_k \label{prob_msr0_o}\\ 
{\rm s.t.} ~&~ {\rm Pr} \{\gamma_k \le \gamma_0\} \le p_0, \forall k, \label{prob_msr0_c1}\\
 ~&~ P_{\rm min} \le p_k \le P_{\rm max}, \forall k, \label{prob_msr0_c2}
\end{align}
\end{subequations}
where $P_{\rm max}$ and $P_{\rm min}$ are the maximum and minimum transmit power budgets of the users, respectively, $\gamma_0$ is the minimum SINR requirement that needs to be maintained by each user to ensure reliability, and the tolerable outage probability is denoted by $p_0$. Note that problem \eqref{prob_msr0} is a non-convex problem with fractional objective and the nonlinear probabilistic constraints, and hence the exactly optimal solution is non-trivial. In the following, we will first develop an acceptable solution to the problem following traditional alternating approaches, and then following ML techniques.

\begin{figure}[ht!]
\centering
\includegraphics[width= 1\linewidth]{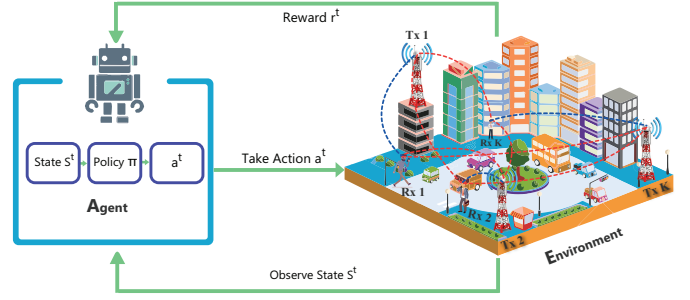} 
\caption{Deep reinforcement learning for IC.
} \label{rl_fig}
\end{figure}

For ease of exposition, let us consider an interference-limited system 
such that the noise term in \eqref{sink} can be ignored. Thus, the SINR in \eqref{sink} reduces to
\begin{align}
 \tilde{\gamma}_{k} \triangleq \frac{\left|h_{k k}\right|^{2} p_{k}}{\sum_{j \neq k}\left|h_{k j}\right|^{2} p_{j} }.\label{sinr_int}
\end{align}
Then the chance SINR outage requirement in \eqref{prob_msr0_c1} can be expressed as
\begin{align}
    {\rm O}_{k} = & {\rm Pr} \{\gamma_k \le \gamma_0\}, \forall k, \nonumber\\ 
    \approx & {\rm Pr} \{\left|h_{kk}\right|^{2} p_{k} \le \gamma_0 \sum_{j\ne k}\left|h_{ kj}\right|^{2} p_{j}\}, \forall k, \label{out_prob}
\end{align}
Note that the outage probability $O_k$ can be interpreted as the fraction of time the $k$th transmitter/receiver pair experiences an outage due to fading. Interestingly, it has been shown analytically in \cite{outage_prob, boyd_out} that the outage probability \eqref{out_prob} can be expressed in closed form using the following Lemma.

\begin{lemma}\label{lem_impli}
If $x_1, x_2 \dots, x_N$ are independent random variables following exponential distribution with means $\mu_i \triangleq \mathbb{E}[x_i] = 1/\lambda_i$, then we have
\begin{align}
{\rm Pr}\left (x_{1} \leq \sum _{i=2}^{N} x_{i} \right) = 1 - \prod _{i=2}^{N} \left({\frac{1}{ 1+\frac{\lambda_1} {\lambda_i}}} \right). \label{lem_posi}
\end{align}
\end{lemma}

\begin{proof}
The proof follows similar lines as in \cite[Appendix~I]{boyd_out}. Therefore, we omit the proof here for brevity.
\end{proof}

Using Lemma~\ref{lem_impli} for an interference-limited scenario, we can express the outage-probability in \eqref{out_prob} as
\begin{align}
    {\rm O}_{k} =1 - \prod _{j \ne k}\frac{1}{1 + \frac{\gamma_0 |{ h}_{kj}|^{2}{p}_j} {|{h}_{kk}|^{2}{P}_k}}.
\end{align}
Thus the outage constraint \eqref{prob_msr0_c1} can be expressed as
\begin{align}
    (1- p_0) \prod _{j \ne k} \left (1+ \frac{\gamma_0 |{h}_{kj}|^{2}{p}_{j}} {|{h}_{kk}|^{2} { p}_{k}} \right) \leq 1. \label{out_posi}
\end{align}
It can be shown that the left-hand side of the inequality \eqref{out_posi} is, in fact, a posynomial function of the powers $p_1, \, p_2, \, \dots, \, p_K$. Thus we can express the problem as
\begin{subequations}\label{prob_msr2}
\begin{align}
\max_{\{p_k\}} ~&~ \min_{k} ~~ R_k  \label{prob_msr2_o}\\
{\rm s.t.} ~&~ (1- p_0) \prod _{j \ne k} \left (1+ \frac{\gamma_0 |{h}_{kj}|^{2}{p}_{j}} {|{h}_{kk}|^{2} {p}_{k}} \right) \leq 1, \forall k, \label{prob_msr2_c1}\\
 ~&~ \frac{P_{\rm min}} {p_k} \le 1, \forall k,  \label{prob_msr2_c2}\\
 ~&~ \frac{p_k}{P_{\rm max}} \le 1, \forall k. \label{prob_msr2_c3}
\end{align}
\end{subequations}
The problem is still non-convex due to the fractional objective function as well as the constraint \eqref{prob_msr2_c1}, however, a solution can be obtained based on the Perron–Frobenius theorem for the maximum eigenvalue of the matrix $\{{\bf H} ~|~ [{\bf H}]_{i,j} = |h_{i,j}|^2\}$ that has non-negative elements (channel gains in this case) \cite{Mitra1994}. Note that due to the monotonicity of the $\log$ function, problem \eqref{prob_msr2} can be equivalently solved by replacing the objective function by the corresponding SINR values. Now, by introducing a slack variable $\eta$, we can reformulate problem \eqref{prob_msr2} as
\begin{subequations}\label{prob_srm3}
\begin{align}
\max_{\eta,\{p_k\}} ~&~ \eta  \label{prob_srm3_o}\\
{\rm s.t.} ~&~  \frac{ \eta \left({\sum_{j\ne k}\left|h_{ kj}\right|^{2} p_{j}}\right)}{{\left|h_{kk}\right|^{2} p_{k}}} \le 1, \forall k, \label{prob_srm3_c0}\\
~&~ (1 - p_0) \prod _{j \ne k} \left (1+ \frac{\gamma_0 |{h}_{kj}|^{2}{p}_{j}} {|{h}_{kk}|^{2} {p}_{k}} \right) \leq 1, \forall k, \label{prob_srm3_c1}\\
 ~&~ \frac{P_{\rm min}} {p_k} \le 1, \forall k,  \label{prob_srm3_c2}\\
 ~&~ \frac{p_k}{P_{\rm max}} \le 1, \forall k. \label{prob_srm3_c3}
\end{align}
\end{subequations}
It can be observed that constraints \eqref{prob_srm3_c0} and \eqref{prob_srm3_c1} are homogeneous with respect to powers thereby depend only on the power ratios. Accordingly, problem \eqref{prob_srm3} is a geometric program (GP) with respect to $\{p_k\}$, therefore, it can be solved globally and efficiently using interior-point methods for geometric programming. We solve the problem using geometric programming and consider it as a baseline to compare the results of our DRL method developed in the next section.

\section{Reinforcement Learning for Resource Optimisation}\label{sec_dnn}
Fig. \ref{rl_fig} demonstrates the basic structure of RL in the system model under consideration. QL  is a  widely used model-free RL technique, which determines the quality of action and informs the RL agent which action will be optimum in each circumstances. 

Lets $S$ denotes possible states of the environment and $A$ denotes a set of discrete actions. Here, $s\in S$ is a specific state containing environmental features. For each time step $t$, agent receives a state $s^t\in S$ from the environment and takes a specific action $a^t$ from the action space $A$. The agent follows a policy  $\pi$ which is the probability of choosing an action $a^t$, given the current state being $s^t$. After the execution of the action, the agent receives a reward $r^t$ and state transition happens from state $s^t$ to state $s^{(t+1)}$. Then the process forms an experience $e^{(t+1)}=\left(s^{(t)}, a^{(t)}, r^{(t)}, s^{(t+1)})\right.$ which summarizes the interaction with the environment. This process is repeated until the agent reaches the terminal state after which the process restarts again. The goal of QL is to find an optimal policy $\pi$ which, at time t, maximizes the future cumulative reward, denoted as \cite{ahmed2019deep}:
\begin{align}
R_{t}=\sum_{k=0}^{\infty} \gamma^{k} r_{t+k}.
\end{align}
Here, $\gamma \in(0,1]$ is a factor which denotes the priority of future rewards in comparison with the current reward. We define an action-value function $Q^\pi(s,a)$, which is the expected return starting from state $s$, taking an action $a$, while following policy $\pi$ \cite{nasir2019multi}:
\begin{align}
Q^{\pi}(s, a)=\mathbb{E}_{\pi}\left[R^{(t)} \mid s^{(t)}=s, a^{(t)}=a\right].
\end{align}
The $Q$-function satisfies the Bellman equation given below:
\begin{align}
  Q^{*} (s, a)=R(s, a)+\gamma \sum_{s^{\prime} \in \mathcal{S}} P_{s s^{\prime}}^{a} \max _{a^{\prime} \in \mathcal{A}} Q^{*}\left(s^{\prime}, a^{\prime}\right).
\end{align}
The QL algorithm uses a lookup table ($Q$-table) to store the output of $Q$-function called $Q$-values. The agent uses an $\epsilon$-greedy policy in order to choose action. According to the $\epsilon$-greedy policy, at each state the agent either chooses to explore the environment by choosing action randomly or exploit $Q$-table and selects the action with the maximum $Q$-value.

The problem with traditional QL is that it stores the $Q$-values in tabular form. So, for high dimensional state spaces, QL becomes infeasible quickly since its space requirement for the $Q$-table becomes impractical. 
In this correspondence, we use DQN which is a combination of traditional QL and a deep neural network (DNN). Instead of using a table to store  the $Q$-values, DQN uses a DNN which receives sates and provides $Q$-values as output for each possible action from that state. To avoid correlation between input data, a buffer of experience called replay memory is also created. The DQN is trained with training data sampled randomly from the replay memory. This technique is known as experience replay. In DQL, we approximate the optimal action-value function by using a neural network, $Q(s, a ; \theta) \approx Q^{*}(s, a).$ Here, $Q(s, a ; \theta)$ is called the DQN and $\theta$ is a parameter of the neural network. 
The MSE of the Bellman equation is reduced by the iterative update, which is used to train the $Q$-network.

\subsection{Proposed DQL Approach}
In our proposed DQL approach, we exploit a single DQN agent which communicates with the environment in order to learn the optimal power of the users. The agent, instead of assigning more power to the user with best channel condition, tries to maintain reliability and assign each user at least minimum power, so that they can maintain the minimum required SINR threshold $\gamma_0$. Once reliability is ensured, the agent gives more emphasis on maximizing the global sum-rate instead of trying to maximize the rate of each individual users. The resource sharing scheme is performed in two phases, first, in the training or learning stage and then in the implementation or testing phase. In the learning phase, the DQN agent access the readily available system performance reward and adjusts its action to achieve an optimal policy by updating the DQN. In implementation phase, the agent receives observation from environment and selects the action according to the $Q$-value received at the output given by the trained DQN.

\subsubsection{State Space}
The state space $S$ consists of power allocation information of other users from the previous step and channel realization of the user ($h_{kj}$), which follows a certain distribution. We assume at the beginning of each time step $t$ that the users' channel information is instantaneously available at the transmitter.

\subsubsection{Action Space}
The resource allocation design using RL comes down to the transmission power control  for each user. So, the total number of actions depend on the number of power levels we are using. Even though transmission power  takes continuous values in most cases for reducing complexity, in this paper,  we limit the power control options to a fixed number of power levels chosen from the interval $[P_{max},P_{min}]$. We assume that the action space $A$ has $n$ power levels.

 \subsubsection{Reward Design}
 Due to inherent flexibility in reward design, RL is an appealing tool for solving problems that are hard to optimize using traditional methods. In the original problem formulation in section II, we set the design objective to maximize the overall sum-rate  ensuring reliability of each user (i.e. $SINR\geq\gamma_0$). Accordingly, we define the reward function for each agent as per the achievable rate of each user, $R_{k}$. So the reward $r_t$ at each time step becomes:
\begin{align}
  r_t = R_{k}.
\end{align} 
To ensure that the reliability constraint is satisfied, the condition in \eqref{prob_srm3_c1} is checked during each reward calculation. If the condition is not fulfilled, the algorithm  simply declares the power choice of that user as invalid and sets the reward to zero. Thus, the overall reward design can be expressed as
 \begin{align}
     R_t=\left\{\begin{array}{ll}
   \!R_{k},\! \!&\! \text { if }\label{reward}
   (1 - p_0) \prod_{j \ne k} \left (1+ \frac{\gamma_0 |{h}_{kj}|^{2}{p}_{j}} {|{h}_{kk}|^{2} {p}_{k}} \right) \leq 1  \\
\!0, \!&\! \text { otherwise. }
\end{array}\right. 
 \end{align}
 
 \subsection{Training the DNN}
The objective of training design is to exploit a trained DQN to replay the experience in learning the resource allocation policy. The DQN takes an observation form the environment $s^t$ as input and outputs the possible actions. Multiple episodes are run to train the DQN and the agent uses the $\epsilon$-greedy policy to explore the environment. After the transition of the environment due to the actions taken as well as the change of channel information, the agent saves the transition tuple $ (s^{(t)}, a^{(t)}, r^{(t)}, s^{(t+1)})$ in a replay memory. A mini batch of transition tuple or experience is sampled randomly from the replay memory. The mini batch is used to train the DQN  and  stochastic  gradient-descent  method is used for updating the DQN. The goal is to minimize the sum-squared error:
\begin{align}
 \sum\left[R^{(t+1)}+\gamma \max _{a^{\prime}} Q\left(s^{(t+1)}, a^{\prime} ;\label{mse} \theta^{-}\right)-Q\left(s^{(t)}, a_{(t)} ; \theta\right)\right]^{2}.  
\end{align}
Here $\theta^{-}$  is the parameter of target $Q$-network which is a DQN  with frozen weights. The target $Q$-network parameters are initialized by duplicating the parameters of training DQN and after a certain number of episodes they are updated to mirror the parameters of the training DQN. The reason for using two DQN is that if we use same $Q$-network for generating the target $Q$-value and the predicted $Q$-value, the learning becomes really unstable. So, 2 DQNs are used to ensure a stable learning. The overall training procedure is summarized in Algorithm~\ref{algo1}.
 
\begin{algorithm} [ht!]
	\caption{Resource Allocation with DRL} 
	 \begin{algorithmic}[1] \label{algo1}
	    \STATE Initialize the environment $E$ (users' channel information)
	    \STATE Initialize the Q-network with action space $A$.
	\FOR {each episode }
	\STATE Reset environment.
	\STATE Initialize $t=0$.
	\REPEAT 
	    \STATE Observe state $S_t$.
     	\STATE Choose action $a_t$ from $A$ using  $\epsilon$-greedy strategy.
    	\STATE Agent takes action $a_t$ and receives reward $r_{t}$.
	    \STATE State transition happens and agent receives next-state $S_{t+1}$.
      	\STATE Store  $(s^{(t)}, a^{(t)}, r^{(t)}, s^{(t+1)})$ in replay memory.
      	\STATE Sample mini-batch from replay memory.
        \STATE Using the stochastic gradient descent method, optimize error between Q-network and learning targets
        defined in \eqref{mse}.
       \STATE t=t+1.
	\UNTIL{ $t>$  maximum step per episode.} 
	\ENDFOR
	\end{algorithmic} 
\end{algorithm}

\subsection{Testing the DNN}
In testing stage, at each time step $t$, the agent receives channel realization $s^t$ and $\epsilon$ is set to the value from the very last training step. The agent selects an action $a^{(t)}$, which has the maximum $Q$-value according to the DQN. After that the agent starts transmission with the power level according to their selected action.
Since the training procedure described in Algorithm~\ref{algo1} is computationally expensive, it can be performed  offline for a large number of channel conditions, over many episodes and the computationally inexpensive, testing or implementation phases can be executed online for the actual network deployment.

\section{Numerical Simulations}\label{sec_sim}
In this section, we compare the analytical results obtained in Section~\ref{sec_vision} with those of the proposed DQN algorithm. The channel coefficients are generated using the standards of Gaussian IC \cite{jrnl_inter}. The channel coefficients follow a Rayleigh fading distribution with zero mean and unit variance. We consider minimum SINR threshold $\gamma_0$ as $-10$ dB and $P_{min}$ as $0$ dB. We use a dense DQN with three hidden layers consisting of $200$, $100$ and $40$ neurons. We also use the rectified linear unit (ReLU), $f(x) = \max(0, x)$ as activation function and \emph{RMSProp} optimizer is used for updating the weights of the network using learning rate of 0.001 and updating the DQN over $3000$ episode. The proposed method is compared with several methods namely: a GP method discussed in Section II, a WMMSE algorithm developed in \cite{wmmse}, a supervised DNN \cite{sun2017learning} and a WMMSE based supervised learning method \cite{JGao19}.

\begin{figure}[ht!]
\centering
\includegraphics[width=0.85\linewidth]{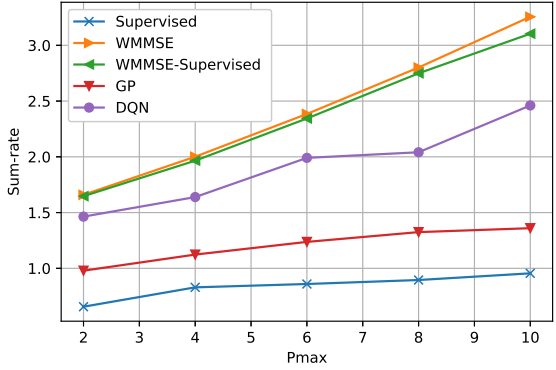}
\caption{Sum-rate versus transmit power budget $P_{\rm max}$.}\label{sum-rate}
\end{figure}

The Fig. \ref{sum-rate} shows the sum rate performance of the DQN in a 4 user scenario with outage probability $p_0 = 0.3$.
We can see that the proposed method performs significantly better than the existing GP algorithm \cite{boyd_out}. The WMMSE approach in \cite{wmmse}  performs better since it does not maintain reliability constraint. Our proposed approach is capable of obtaining better result than the approach in \cite{wmmse}, if we ignore the reliability constraint during reward design as shown in Fig.  \ref{sum-rate}.

\begin{figure}[ht!]
\centering
\includegraphics[width=0.85\linewidth]{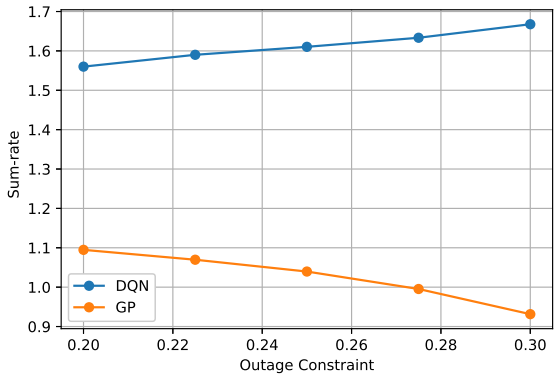}
\caption{Sum-rate versus outage probability threshold $p_0$.}\label{p0}
\end{figure}
Fig. \ref{p0} shows the performance of our proposed method over varying outage probability $p_0$.  We make the comparison considering $P_{max}$ as $4$ dB. It can be seen from the figure that the proposed DQN method significantly outperforms traditional GP method. It can also be observed that the sum-rate of the proposed method increases with $p_0$, while in traditional approaches performance decreases. This can be explained with the reward design \eqref{reward}. When $p_0$ increases, due to the term $(1-p_0)$ the overall condition term decreases and thus offers higher probability of becoming less then or equal to one and full-filling the condition for receiving reward resulting in increased sum-rate performance. 

\begin{figure}[ht!]
\centering
\includegraphics[width=0.85\linewidth]{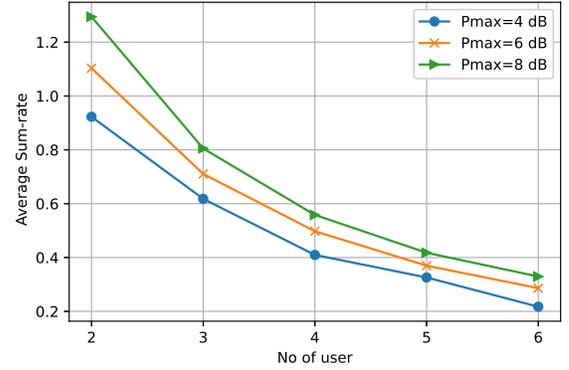}
\caption{Average sum-rate versus the number of user pairs $K$.}\label{p1}
\end{figure}

Fig. \ref{p1} shows average sum rate against different number of users. As expected, the average  sum-rate performance of the DQN decreases as the number of user increases in the system since higher number of users generate additional interference. This effect can be compensated by increasing transmit power budget as evident from the Fig. \ref{p1}.

\section{Conclusion}\label{sec_con}
 In this paper, we have developed a DRL based resource allocation scheme which successfully learns the power allocation of a wireless interference channel and guarantees reliability requirements for each user. We exploit a DQN with experience to learn the relationship between users' channel information and power allocation. We have demonstrated the superiority of the proposed approach by comparing the sum-rate performance with that of the existing GP-based solution.

\label{sec4}

\ifCLASSOPTIONcaptionsoff
  \newpage
\fi

\bibliographystyle{IEEEtran}\footnotesize{

\bibliography{IEEEabrv,refdb}}%

\end{document}